\documentclass[twocolumn,superscriptaddress,
showpacs,preprintnumbers,amsmath,amssymb]{revtex4}
\usepackage{hyperref}

\usepackage{amssymb,amsmath,amsthm}
\usepackage{epsfig}
\usepackage{multirow}
\newtheorem{Thm}{Theorem}
\newtheorem{Cor}{Corollary}
\newtheorem{Lem}[Thm]{Lemma}

\newtheorem{definition}{Definition}
\theoremstyle{definition}

\newtheorem{Rmk}{Remark}

\newcommand{\bra}[1]{{\left\langle #1 \right|}}
\newcommand{\ket}[1]{{\left| #1 \right\rangle}}

\newcommand{\C}{\mbox{$\mathbb C$}}
\newcommand{\R}{\mbox{$\mathbb R$}}

\newcommand{\T}{\mbox{$\mathrm{tr}$}}

\begin{document}
\title{Polygamy of Distributed Entanglement}

\author{Francesco Buscemi}
\affiliation{Statistical Laboratory, DPMMS, University of Cambridge,
  Wilberforce Road, CB3 0WB, UK }
\author{Gilad Gour}
\affiliation{
Institute for Quantum Information Science, University of Calgary, Alberta T2N 1N4, Canada
}
\affiliation{
Department of Mathematics and Statistics, University of Calgary
, Alberta T2N 1N4, Canada
}
\author{Jeong San Kim}\email{jkim@qis.ucalgary.ca}
\affiliation{
 Institute for Quantum Information Science,
 University of Calgary, Alberta T2N 1N4, Canada
}

\date{\today}
%
\begin{abstract}
  While quantum entanglement is known to be monogamous (i.e. shared
  entanglement is restricted in multi-partite settings), here we show
  that distributed entanglement (or the potential for entanglement) is
  by nature polygamous. By establishing the concept of one-way
  unlocalizable entanglement (UE) and investigating its properties, we
  provide a polygamy inequality of distributed entanglement in
  tripartite quantum systems of arbitrary dimension. We also provide a
  polygamy inequality in multi-qubit systems, and several trade offs
  between UE and other correlation measures.
\end{abstract}

\maketitle

\section{Introduction}
Quantum entanglement is a non-local quantum correlation providing
a lot of useful applications in the field of quantum
communications and computations such as quantum teleportation and
quantum key distribution~\cite{tele, qkd1, qkd2}. This important
role of quantum entanglement has stimulated intensive study in
both way of its quantification and qualification.

One of the essential differences of quantum correlations (especially,
quantum entanglement) from other classical ones is that it cannot be
freely shared among the parties in multipartite quantum systems. In
particular, a pair of components that are maximally entangled cannot
share entanglement~\cite{CKW,ov} nor classical correlations~\cite{KW}
with any part of the rest of the system, hence the term {\em Monogamy
  of Entanglement}~(MoE)~\cite{T04}. Monogamy of entanglement was
shown to have a complete mathematical characterization for multi-qubit
systems~\cite{ov} using a certain entanglement measures, the
concurrence~\cite{ww}.

Whereas MoE shows the restricted sharability of multi-party quantum
entanglement, the distribution of entanglement, or {\em Entanglement
  of Assistance}~(EoA)~\cite{d, cohen} in multipartite quantum systems
was shown to have a dually monogamous (or {\em Polygamous})
property. Using {\em Concurrence of Assistance}~(CoA)~\cite{lve} as
the measure of distributed entanglement, it was also shown that
whereas monogamy of entanglement inequalities provide an upper bound
for bipartite sharability of entanglement in a multipartite system,
the same quantity provides a lower bound for distribution of bipartite
entanglement in a multipartite system~\cite{gbs}. In this paper, by
introducing the concept of {\em One-way Unlocalizable
  Entanglement}~(UE), we provide a polygamy inequality of entanglement
in tripartite quantum systems of arbitrary dimension using entropic
entanglement measure. Based on the functional relation between
concurrence and entropic measure in two-qubit systems, we provide a
polygamy inequality in multi-qubit systems. We also provide several
trade offs between UE and other correlations such as EoA, and
localizable entanglement.

The paper is organized as follows. In Sec.~\ref{sec: qf}, we provide
the definition of UE, and its basic properties. In Sec.~\ref{sec:
  3dual}, we provide a polygamy inequality of distributed entanglement
in tripartite quantum systems in terms of entropy and EoA. In
Sec.~\ref{sec: ndual}, we generalize the polygamy inequality of
entanglement into multi-qubit systems, and provide a more tight
polygamy inequality for three-qubit systems. In Sec.~\ref{sec: other},
we provide several trade offs between UE and other correlations, and
we summarize our results, in Sec.~\ref{sec: Conclusion}.

\section{One-Way Unlocalizable Entanglement}
\label{sec: qf}
\subsection{Definition}
\label{subsec: def}
For any bipartite quantum state $\rho_{AB}$,
its one-way distillable common randomness~\cite{DV}
is defined as
\begin{equation}
  C_D^\leftarrow(\rho_{AB})=\lim_{n\rightarrow \infty}
                              \frac{1}{n}I^\leftarrow(\rho_{AB}^{\otimes n}),
\end{equation}
where, the function $I^\leftarrow(\rho_{AB})
$~\cite{Henderson-Vedral01} is
\begin{equation}\begin{split}
  I^\leftarrow(\rho_{AB}) &=  \max_{\{M_x\}} \left[S(\rho_A)-\sum_x p_x S(\rho^x_A)\right],
\end{split}
\label{HVI}
\end{equation}
and where the maximum is taken over all the measurements $\{M_x\}$
applied on system $B$.  Here, $S(\rho_A)$ is the von Neumann entropy
of $\rho_A\equiv \T_B(\rho_{AB})$, $p_x\equiv \T[(I_A\otimes
M_x)\rho_{AB}]$ is the probability of the outcome $x$, and
$\rho^x_A\equiv \T_B[(I_A\otimes {M_x})\rho_{AB}]/p_x$ is the state of
system $A$ when the outcome was $x$.

For a tripartite pure state $\ket{\psi}_{ABC}$ with
$\rho_{A}=\T_{BC}\ket{\psi}_{ABC}\bra{\psi}$, $\rho_{AB}
=\T_{C}\ket{\psi}_{ABC}\bra{\psi}$, and $\rho_{AC}
=\T_{B}\ket{\psi}_{ABC}\bra{\psi}$, it was shown that~\cite{KW}
\begin{align}
    S(\rho_A)&=I^\leftarrow(\rho_{AB})+E_f(\rho_{AC}).
    \label{KWmain1}
\end{align}
Here, $E_f(\rho_{AC})$ is the {\em Entanglement of
Formation}~(EoF) of $\rho_{AC}$ defined as~\cite{bdsw}
\begin{equation}
E_f(\rho_{AC})=\min \sum_{i}p_i S(\rho^{i}_{A}),
\label{eof}
\end{equation}
where the minimization is taken over all pure state decomposition of
$\rho_{AC}$ such that,
\begin{equation}
\rho_{AC}=\sum_{i} p_i |\phi^i\rangle_{AC}\langle\phi^i|,
\label{decomp}
\end{equation}
with $\T_{C}|\phi^i\rangle_{AC}\langle\phi^i|=\rho^{i}_{A}$.


As a dual quantity to EoF, EoA is defined by the maximum average
entanglement of $\rho_{AC}$,
\begin{equation}
E_a(\rho_{AC})=\max \sum_{i}p_i S(\rho^{i}_{A}),
\label{EoA}
\end{equation}
over all possible pure state decompositions of $\rho_{AC}$.

\begin{definition} The \emph{one-way unlocalizable entanglement}~(UE)
  of a bipartite state $\rho_{AB}$ is defined as follows:
\begin{equation}
  E_u^{\leftarrow}(\rho_{AB}) := S(\rho_A)-E_a(\rho_{AC}),
  \label{puremain1}
\end{equation}
where $\rho_{AC}$ denotes the reduced state of a purification
$\ket{\psi}_{ABC}$ of $\rho_{AB}$.
\end{definition}

The one-way unlocalizable entanglement can be equivalently
characterized as follows:

\begin{Lem}
  For any given bipartite state $\rho_{AB}$, its one-way unlocalizable
  entanglement is given by
\begin{equation}
  E_u^{\leftarrow}(\rho_{AB})=\min_{\{M_x\}}\left[S(\rho_A)-\sum_xp_xS(\rho_A^x)\right],
\label{fragility}
\end{equation}
where the minimum is taken over all possible rank-1 measurements
$\{M_x\}$ applied on subsystem $B$.
\label{Lem: puremain1}
\end{Lem}
\begin{proof}
Eq.~(\ref{fragility}) can be rewritten as
\begin{equation}
E_u^{\leftarrow}(\rho_{AB}) =  S(\rho_A)- \max_{\{M_x\}}\sum_x p_x S(\rho^x_A),
\label{unlocal2}
\end{equation}
where the maximum is taken over all possible rank-1 measurements
$\{M_x\}$ applied on system $B$.

Since $\ket{\psi}_{ABC}$ is a pure state, all possible pure state
decompositions of $\rho_{AC}$ can be realized by rank-1 measurements
of subsystem $B$, and conversely, any rank-1 measurement can be
induced from a pure state decomposition of $\rho_{AC}$. Thus, the
second term on the right hand side of Eq.~(\ref{unlocal2}) is the
maximum average entanglement over all possible pure state
decomposition of $\rho_{AC}$, which is the definition of
$E_a(\rho_{AC})$, and this completes the proof.
\end{proof}

By definition, the UE of $\rho_{AB}$ is the difference between
$S(\rho_A)$ and $E_a(\rho_{AC})$. Here, $S(\rho_A)$ quantifies the
entanglement of the pure state $\ket{\psi}_{A(BC)}$ with respect to
the bipartite cut $A$--$BC$, whereas $E_a(\rho_{AC})$ measures the
maximum average entanglement that can be localized on the subsystem
$AC$ with the assistance of $B$. The terminology used is then
clear. Figure~\ref{relation} graphically illustrates this separation.

\begin{figure}
 \includegraphics[width=6.5cm]{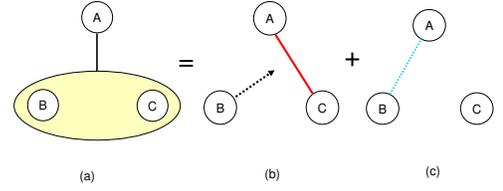}\\
  \caption{The entanglement of $\ket{\psi}_{A(BC)}$ ((a) in the figure)
   consists of two different parts: One is the localizable entanglement onto $AC$
   represented by $E_a(\rho_{AC})$ ((b) in the figure), and the other part
   is the unlocalizable entanglement represented by $E_u^{\leftarrow}(\rho_{AB})$ ((c) in the figure). }\label{relation}
\end{figure}


\subsection{Properties}
\label{subsec: prop}

\subsubsection{Subadditivity}
\label{subsub: subadditivity}
\begin{Lem}
For all bipartite states $\rho_{AB}$ and $\sigma_{A'B'}$,
\begin{equation}
E_u^{\leftarrow}(\rho_{AB}\otimes
\sigma_{A'B'})\leq E_u^{\leftarrow}(\rho_{AB})+E_u^{\leftarrow}(\sigma_{A'B'}),
\label{subadd}
\end{equation}
where
\begin{equation}
\begin{split}
E_u^{\leftarrow}(\rho_{AB}\otimes \sigma_{A'B'})= \min_{\{L_z\}}
 \left[S(\rho_A \otimes \sigma_{A'})-\sum_z r_z S(\tau^z_{AA'})\right],\\
\end{split}
\end{equation}
with $r_z=\T[(I_{AA'}\otimes L_{z})\rho_{AB}\otimes
\sigma_{A'B'}]$, $\tau^z_{AA'}=\T_{BB'}[(I_{AA'}\otimes
L_{z})\rho_{AB}\otimes \sigma_{A'B'}]/r_{z}$, and the minimum is
taken over all possible rank-1 measurements $\{L_z\}$ applied on
subsystem $BB'$. \label{lem: subadd}
\end{Lem}

\begin{proof}
Let $\{M_x \}$ and $\{N_y \}$ be the optimal rank-1 measurements
on subsystems $B$ and $B'$ for $E_u^{\leftarrow}(\rho_{AB})$ and
$E_u^{\leftarrow}(\sigma_{A'B'})$ respectively, then, we have

\begin{align}
E_u^{\leftarrow}&(\rho_{AB})+E_u^{\leftarrow}(\sigma_{A'B'})\nonumber\\
&=S(\rho_A)+S(\sigma_{A'})-\sum_{x}p_x S(\rho^{x}_{A})-\sum_{y}q_y S(\sigma^{y}_{A'})\nonumber\\
&=S(\rho_{A}\otimes \sigma_{A'})-\sum_{xy}p_x q_y S(\rho^x_{A}\otimes \sigma^y_{A'})\nonumber\\
&\geq E_u^{\leftarrow}(\rho_{AB}\otimes \sigma_{A'B'}),
\end{align}
where $p_x\rho^x_A=\T_{B}[(I_A\otimes M_{x})\rho_{AB}]$,
$q_y\sigma^y_{A'}=\T_{B'}[(I_{A'}\otimes N_{y})\sigma_{A'B'}]$,
and the second equality is due to the additivity of von Neumann
entropy and the definition of
$E_u^{\leftarrow}(\rho_{AB}\otimes \sigma_{A'B'})$.
\end{proof}

By Lemma~\ref{lem: subadd}, we can assure the existence of the
{\em regularized} UE
\begin{equation}
E_{u,\infty}^\leftarrow(\rho_{AB}):=\lim_{n\rightarrow \infty}\frac{E_u^{\leftarrow}(\rho_{AB}^{\otimes n})}{n},
\label{regular}
\end{equation}
which satisfies
\begin{equation}
E_{u,\infty}^\leftarrow(\rho_{AB})\leq E_u^{\leftarrow}(\rho_{AB}).
\label{regular2}
\end{equation}

\subsubsection{Simple Lower Bound}
\label{subsub: bounds}
\begin{Lem}
For any bipartite state $\rho_{AB}$,
\begin{equation}
E_u^{\leftarrow}(\rho_{AB})\geq \max\{I_{c}^{\leftarrow}(\rho_{AB}),~ 0 \},
\label{lower}
\end{equation}
where $I_{c}^{\leftarrow}(\rho_{AB}):=S(\rho_A)-S(\rho_{AB})$ is
the coherent information of $\rho_{AB}$. \label{lem: lower}
\end{Lem}
\begin{proof}
Let $\ket{\psi}_{ABC}$ be a purification of $\rho_{AB}$, then due
to the monotonicity of entanglement, we have
\begin{equation}
E_a(\rho_{AC})\leq \min\{S(\rho_A),~S(\rho_C)\},
\label{EoAupper}
\end{equation}
where $\rho_{AC}=\T_{B}\ket{\psi}_{ABC}\bra{\psi}$.

Thus, together with Lemma~\ref{Lem: puremain1}, we have
\begin{align}
E_u^{\leftarrow}(\rho_{AB})&=S(\rho_A)-E_a(\rho_{AC})\nonumber\\
&\geq \max\{S(\rho_A)-S(\rho_C),~0\}\nonumber\\
&=\max\{I_{c}^{\leftarrow}(\rho_{AB}),~ 0 \},
\end{align}
where the last equality is due to the purity of $\ket{\psi}_{ABC}$, that is,
$S(\rho_C)=S(\rho_{AB})$.
\end{proof}

Since $\ket{\psi}_{ABC}^{\otimes n}$ is a purification of both
$\rho_{AB}^{\otimes n}$ and $\rho_{BC}^{\otimes n}$, we have
\begin{equation}
  E_u^{\leftarrow}\left(\rho_{AB}^{\otimes n}\right)+E_a\left(\rho_{AC}^{\otimes n}\right)=nS(\rho_A).
\end{equation}
By taking the limit $n\to\infty$, and due to the
relation~\cite{Smolin-Ver-Win}
\begin{equation}
\lim_{n\to\infty}\frac{E_a\left(\rho_{AC}^{\otimes n}\right)}{n}=\min\{S(\rho_A),S(\rho_C)\},
\end{equation}
we have that
\begin{equation}
{E_{u,\infty}^{\leftarrow}}(\rho_{AB})=\max\{I_c^{\leftarrow}(\rho_{AB}),0\}.
\label{asymp}
\end{equation}

Eq.~(\ref{asymp}) implies that, in the asymptotic limit of many
copies, separable states do not exhibit quantumness in their
correlations, or their correlations are completely erasable. This
is a strong evidence that the distinction between separable and
entangled states is {\em operational} only in asymptotic sense,
since separable states can exhibit non-zero UE in finite
case.

\section{Polygamy of entanglement in tripartite quantum systems}
\label{sec: 3dual}
For any bipartite pure state $\ket \phi_{AB}$, its concurrence, $\mathcal{C}(\ket
\phi_{AB})$ is defined as~\cite{ww}
\begin{equation}
\mathcal{C}(\ket \phi_{AB})=\sqrt{2(1-\T\rho^2_A)},
\label{pure state concurrence}
\end{equation}
where $\rho_A=\T_B(\ket \phi_{AB}\bra \phi)$.
For any mixed state $\rho_{AB}$, its concurrence is defined via {\em convex-roof extension}, that is,
\begin{equation}
\mathcal{C}(\rho_{AB})=\min \sum_k p_k \mathcal{C}({\ket {\phi_k}}_{AB}),
\label{mixed state concurrence}
\end{equation}
where the minimum is taken over all possible pure state
decompositions, $\rho_{AB}=\sum_kp_k{\ket {\phi_k}}_{AB}\bra
\phi_k$.

As a dual value to concurrence, CoA~\cite{lve} of $\rho_{AB}$ is defined as
\begin{equation}
\mathcal{C}^a(\rho_{AB})=\max \sum_k p_k \mathcal{C}({\ket {\phi_k}}_{AB}),
\label{CoA}
\end{equation}
where the maximum is taken over all possible pure state decompositions of
$\rho_{AB}$.

By using concurrence and CoA as the quantification of bipartite entanglement, it was
shown that there exists a polygamy relation of
entanglement in multi-qubit systems~\cite{gbs}. More precisely,
for any pure state $\ket{\psi}_{A_1 \cdots A_n}$ in an $n$-qubit
system $\mathcal H_{A_1} \otimes \cdots \otimes \mathcal H_{A_n}$
where $\mathcal H_{A_i} \cong
\mathbb{C}^2$ for $i=1,\ldots,n$,
\begin{equation} \mathcal{C}_{A_1 (A_2 \cdots
A_n)}^2  \leq  (\mathcal{C}^a_{A_1 A_2})^2
+\cdots+(\mathcal{C}^a_{A_1 A_n})^2, \label{ndual}
\end{equation}
where $\mathcal{C}_{A_1 (A_2 \cdots A_n)}$ is the concurrence of
$\ket{\psi}_{A_1 \cdots A_n}$ with respect to the bipartite cut $A_1$
and $A_{2}\cdots A_{n}$, and $\mathcal{C}^a_{A_1 A_i}$ is the CoA of
$\rho_{A_1A_i}=\T_{A_2\cdots A_{i-1}A_{i+1}\cdots A_n}\left(\ket{\psi}_{A_1 \cdots A_n}\bra{\psi}\right)$ for $i=2,
\ldots ,n$.

In this section, we provide an analytic upper bound of UE
in Eq.~(\ref{fragility}), and derive a polygamy inequality of
entanglement in terms of von-Neumann entropy and EoA
for tripartite quantum systems of arbitrary dimension.

First, for an upper bound of UE, we have the following
theorem.

\begin{Thm}
For any bipartite state $\rho_{AB}$ in a bipartite quantum system
$\mathcal{H}_A \otimes \mathcal{H}_B$,
\begin{equation}
E_u^{\leftarrow}(\rho_{AB})\leq\frac{I(\rho_{AB})}{2}, \label{upper}
\end{equation}
where $I(\rho_{AB})=S(\rho_A)+S(\rho_B)-S(\rho_{AB})$ is the
mutual information of $\rho_{AB}$. \label{Thm: upper}
\end{Thm}

\begin{proof}
  Let $\rho_B=\sum_{i=0}^{d_B-1}\lambda_{i}\ket{e_i}_B\bra{e_i}$ be a
  spectral decomposition of $\rho_B=\T_{A}(\rho_{AB})$ where $d_B$ is
  the dimension of the subsystem $\mathcal{H}_B$. The proof method
  follows the construction used in~\cite{christandl}.

For any state $\sigma \in \mathcal{H}_B$, define the channels
\begin{align}
M_{0}(\sigma):&=\sum_{i=0}^{d_B-1}\ket{e_i}\bra{e_i}\sigma\ket{e_i}\bra{e_i}\nonumber\\
M_{1}(\sigma):&=\sum_{i=0}^{d_B-1} |\tilde e_j\rangle\langle\tilde
e_j|\sigma|\tilde e_j\rangle\langle\tilde e_j|, \label{channels1}
\end{align}
where $\{ |\tilde e_j \rangle \}_{j} $ is the {\em Fourier basis}
such that,
\begin{equation}
|\tilde e_j \rangle = \frac{1}{\sqrt{d}}\sum_{k=0}^{d_B-1}
\omega_d^{jk}\ket{e_k},~j=0,\ldots ,d_B-1, \label{fourier}
\end{equation}
and $\omega_d = e^{\frac{2\pi i}{d}}$ is the $d$-th root of unity.

Notice that $M_0(\rho_B)=\rho_B$, and
$M_1(\rho_B)=\frac{1}{d_B}I_B$, so that
$M_1(M_0(\rho_B))=M_0(M_1(\rho_B))=\frac{1}{d_B}I_B$. We can also
write
\begin{equation}
M_0(\sigma)=\frac{1}{d_B}\sum_{b=0}^{d_B-1}Z^b\sigma Z^{-b},~
M_{1}(\sigma)=\frac{1}{d_B}\sum_{a=0}^{d_B-1}X^a\sigma X^{-a},
\label{channels2}
\end{equation}
where $Z$ and $X$ are generalized $d_B$-dimensional Pauli operators,
\begin{align}
Z=&\sum_{j=0}^{d_B-1}\omega_d^j\ket{e_j}\bra{e_j},\nonumber\\
X=&\sum_{j=0}^{d_B-1}\ket{e_{j+1}}\bra{e_j}=\sum_{j=0}^{d_B-1} \omega_d^{-j}|\tilde
e_j \rangle \langle \tilde e_j |.\label{paulis}
\end{align}

In the following, we will write
\begin{align}
(I_A\otimes M_0)(\rho_{AB})&=\sum_{i=0}^{d_B-1} \sigma_A^i \otimes \lambda_i\ket{e_i}_B\bra{e_i},\nonumber\\
(I_A\otimes M_1)(\rho_{AB})&=\sum_{j=0}^{d_B-1} \tau_A^j \otimes
\frac{1}{d_B}|\tilde e_j \rangle_B \langle \tilde e_j|,
\end{align}
where $\lambda_i\sigma_A^i=\T_B [(I_A \otimes
\ket{e_i}_B\bra{e_i})\rho_{AB}]$, and $\tau_A^j/d_B=\T_B[(I_A
\otimes |\tilde e_j\rangle_B \langle\tilde e_j|)\rho_{AB}]$ for $i,~j \in \{0,\cdots, d_B-1\}$.

The induced ensembles on $A$ by the channels $M_0$ and $M_1$ will be denoted by $\mathcal
E_0:=\{\lambda_i,\sigma_A^i\}_i$ and $\mathcal E_1:=\{\frac
{1}{d_B},\tau_A^j\}_j$, and the entropy defects of the induced ensembles
on $A$ will be denoted as
\begin{align}
\chi(\mathcal E_0)=&S(\rho_A)-\sum_{i=0}^{d_B-1}\lambda_i S(\sigma_A^i),\nonumber\\
\chi(\mathcal E_1)=&S(\rho_A)-\frac{1}{d_B}\sum_{i=0}^{d_B-1}S(\tau_A^j).
\label{chi}
\end{align}

By defining a four-partite quantum state $\Omega_{XYAB}$ in
$\mathcal B \left( \C^{d_B} \otimes \C^{d_B} \otimes \C^{d_A} \otimes \C^{d_B}\right)$
such that
\begin{widetext}
\begin{equation}
  \Omega_{XYAB}:=\frac 1{d^2_B}\sum_{x,y=0}^{d_B-1}\ket{x}_X
  \bra{x}\otimes\ket{y}_Y\bra{y}\otimes(I_A\otimes X^x_BZ^y_B)\rho_{AB}(I_A\otimes
  Z^{-y}_BX^{-x}_B),
\label{XYAB}
\end{equation}
we have
\begin{align}
\Omega_{XAB}=&\frac
1{d_B}\sum_{x=0}^{d_B-1}\ket{x}_X\bra{x}\otimes X^x_B
\left(\sum_{i=0}^{d_B-1} \sigma_A^i \otimes \lambda_i\ket{e_i}_B\bra{e_i}\right)X_B^{-x},\nonumber\\
\Omega_{YAB}=&\frac
1{d_B}\sum_{y=0}^{d_B-1}\ket{y}_Y\bra{y}\otimes
Z_B^y\left(\sum_{j=0}^{d_B-1} \tau_A^j \otimes
\frac{1}{d_B}|\tilde e_j \rangle_B \langle \tilde
e_j|\right)Z_B^{-y}, \label{XAB, YAB}
\end{align}
\end{widetext}
and
\begin{equation}
  \Omega_{AB}=\rho_A\otimes\frac{I_B}{d_B}.
\end{equation}
By straightforward calculation, we can obtain
\begin{align}
I(\Omega_{X(AB)})=&S(\Omega_X)+S(\Omega_{AB})-S(\Omega_{XAB})\nonumber\\
=&\log d_B+\log d_B+S(\rho_A)-\log d_B\nonumber\\
    &~~~~~~~~~~~~- S\left(\sum_i \sigma_A^i \otimes \lambda_i\ket{e_i}_B\bra{e_i}\right)\nonumber\\
=&\log d_B+S(\rho_A)-H(\vec\lambda)-\sum_i\lambda_iS(\sigma_A^i)\nonumber\\
=&\log d_B-S(\rho_B)+\chi(\mathcal E_0), \label{Ixab}
\end{align}
where $I(\Omega_{X(AB)})$ is the mutual information of $\Omega_{X(AB)}$ with respect to the bipartite cut
$X-AB$, and the second, third equalities are due to the {\em joint
entropy theorem}~\cite{nc}. Analogously, we have
\begin{align}
&I(\Omega_{Y(AB)})=\chi(\mathcal E_1),\nonumber\\
&I(\Omega_{(XY)(AB)})=\log d_B+S(\rho_A)-S(\rho_{AB})\nonumber\\
&~~~~~~~~~~~~~~~~~~=\log d_B+I_c^{\leftarrow}(\rho_{AB}). \label{Iyab Ixyab}
\end{align}

Due to the independence of subsystems $X$ and $Y$, we have
$I(\Omega_{(XY)(AB)})\geq I(\Omega_{X(AB)})+I(\Omega_{Y(AB)})$,
which implies
\begin{equation}
\chi(\mathcal E_0)+\chi(\mathcal E_1)\leq I(\rho_{AB}).
\label{ensembleineq}
\end{equation}

Since $\chi(\mathcal E_0)$ and $\chi(\mathcal E_1)$ of
Eq.~(\ref{ensembleineq}) can be obtained, respectively, from
$\rho_{AB}$ by rank-1 measurements $\{\ket{e_i}_B\bra{e_i} \}_i$
and $\{ |\tilde e_j \rangle_B \langle \tilde e_j| \}_j$ of
subsystem $B$, by defining a rank-1 measurement
\begin{equation}
  \left\{ \frac{\ket{e_i}_B\bra{e_i}}{2}, \frac{|\tilde e_j
    \rangle_B \langle \tilde e_j|}{2}\right\}_{i,j}, \label{povm}
\end{equation}
we have
\begin{equation}
E_u^{\leftarrow}(\rho_{AB})\leq  \frac{\chi(\mathcal
E_0)}{2}+\frac{\chi(\mathcal E_1)}{2}\leq \frac{I(\rho_{AB})}{2},
\label{uppermain}
\end{equation}
which completes the proof.
\end{proof}

\begin{Cor}
For any tripartite pure state $\ket{\psi}_{ABC}$, we have
\begin{equation}
S(\rho_A)\leq E_a(\rho_{AB})+E_a(\rho_{AC}). \label{poly}
\end{equation}
\label{Cor: dual}
\end{Cor}
\begin{proof}
By Lemma~\ref{puremain1}, we have
\begin{align}
E_a(\rho_{AC})&=S(\rho_A)- E_u^{\leftarrow}(\rho_{AB}) ,\nonumber\\
E_a(\rho_{AB})&=S(\rho_A)- E_u^{\leftarrow}(\rho_{AC}),
\end{align}
and thus,
\begin{equation}
E_a(\rho_{AC})+E_a(\rho_{AB})=2S(\rho_A)-
E_u^{\leftarrow}(\rho_{AB}) - E_u^{\leftarrow}(\rho_{AC}).
\end{equation}
Now, by Theorem~\ref{Thm: upper}, we have
\begin{align}
E_a(\rho_{AC})+& E_a(\rho_{AB})\nonumber\\
 \geq & 2S(\rho_A)-\frac{I(\rho_{AB})}{2}-\frac{I(\rho_{AC})}{2}\nonumber\\
=&2S(\rho_A)-S(\rho_A)/2-S(\rho_B)/2+S(\rho_{AB})/2 \nonumber\\
 & -S(\rho_A)/2-S(\rho_C)/2+S(\rho_{AC})/2\nonumber\\
=&S(\rho_A).
\end{align}
\end{proof}

Corollary~\ref{Cor: dual} tells us that for a tripartite pure
state of arbitrary dimension, there exists a polygamy relation of
entanglement in terms of entropy of entanglement and EoA.
Furthermore, this is, we believe, the first result of the
polygamous (or dually monogamous) property of distribution of entanglement in
multipartite higher-dimensional quantum systems rather than qubits.

\section{Polygamy relation of entanglement in multi-qubit quantum systems}
\label{sec: ndual}

In this section, we show that the polygamy inequality of
entanglement in Corollary~\ref{Cor: dual} can be generalized into
multipartite quantum systems for the case when each subsystem is a
two-level quantum system. By investigating the functional relation
between concurrence and EoF in two-qubit systems~\cite{ww}, we
show that there exists a polygamy inequality of entanglement in
terms of entropy and EoA in $n$-qubit systems. We also show that,
in three-qubit systems, we have a more tight polygamy inequality
than Eq.~(\ref{poly}) in Corollary~\ref{Cor: dual}.

First, let us consider the functional relation of concurrence with
EoF in two-qubit systems.  For a 2-qubit mixed state $\rho_{AB}$
(or a pure state $\ket{\psi}_{AB} \in \mathbb{C}^2 \otimes
\mathbb{C}^{d}$), the relation between its concurrence,
$\mathcal{C}_{AB}$ and $E_{f}(\rho_{AB})$ can be given as a
monotone increasing, convex function $\mathcal{E}$~\cite{ww}, such
that
\begin{equation}
 E_f (\rho_{AB}) = {\mathcal E}(\mathcal{C}_{AB}),
\end{equation}
where
\begin{equation} {\mathcal E}(x) = H\Bigl({1\over 2} + {1\over
2}\sqrt{1-x^2}\Bigr), \hspace{0.5cm}\mbox{for } 0 \le x \le 1,
\label{eps}
\end{equation}
and $H(\cdot)$ is the binary entropy function $H(x) = -[x\log_2 x +
(1-x)\log_2 (1-x)]$. The same function $\mathcal{E}(x)$ relates also
the EoA of a bipartite state $\rho_{AB}$ with its CoA
via the equation
\begin{equation}
  E_a(\rho_{AB})\ge\mathcal{E}(\mathcal{C}_{AB}^a),
\label{eoacoa}
\end{equation}
which is due to the convexity of ${\mathcal E}$
and the definition of EoA.
The following lemma shows an important property of the function
$\mathcal{E}(x)$.
\begin{Lem}
\begin{equation}
{\mathcal E}(\sqrt{x^2 + y^2}) \leq {\mathcal E}(x) + {\mathcal
E}(y), \label{Epro}
\end{equation}
for $0 \leq x,~y\leq 1$ such that $0\leq x^2+y^2\leq 1$.
\label{Lem: E}
\end{Lem}
\begin{proof}
By considering \begin{equation} f(x,y)= {\mathcal E}(x) +
{\mathcal E}(y)-{\mathcal E}(\sqrt{x^2 + y^2}), \label{f}
\end{equation} as a two-vairable real-valued function on the
domain $D=\{(x,y)| 0 \leq x,~y \leq 1, 0 \leq \ x^2 +  y^2 \leq 1
\}$, it is enough to show that $f(x,y)\geq 0$ in $D$.

Since $D$ is a compact subset in $\R^2$, whereas $f$ is analytic
on the interior of $D$, and continuous on $D$, the minimum value
of $f$ arises only on the critical points or on the boundary of
$D$. It can be directly checked that $f$ does not have any
vanishing gradient on the interior of $D$, and has non-negative
function values on the boundary of $D$. Thus, $f$ is non-negative
on the domain $D$.
\end{proof}

\subsection{Three-qubit systems}
\label{subsec: 3-qubit}

A direct observation from~\cite{CKW} shows that, for a 3-qubit
pure state $\ket{\psi}_{ABC}$,
\begin{equation}
 \mathcal{C}_{A(BC)}^2 = \mathcal{C}_{AB}^2 + ({\mathcal{C}^{a}_{AC}})^2,
\label{3tangle}
\end{equation}
where $\mathcal{C}_{AB}$ and $\mathcal{C}^{a}_{AC}$ are the
concurrence and concurrence of assistance of $\rho_{AB}$ and
$\rho_{AC}$ respectively. (Later, Eq.~(\ref{3tangle}) was formally
shown in~\cite{ys}.) From Eq.~(\ref{3tangle}) together with
Lemma~\ref{Lem: E}, we have the following theorem.
\begin{Thm}
For a three-qubit pure state $\ket{\psi}_{ABC}$,
\begin{equation}
S(\rho_{A}) \leq  E_f(\rho_{AB}) + E_a(\rho_{AC}).
\label{3qubitineq}
\end{equation}
\label{thm2}
\end{Thm}

\begin{proof}
Since $\ket{\psi}_{ABC}$ is a bipartite pure state in
$\mathbb{C}^2 \otimes \mathbb{C}^{4}$ with respect to the
bipartite cut $A$ and $BC$, we have,
\begin{equation}
S(\rho_{A})= E_f(\psi_{A(BC)})={\mathcal E}(\mathcal{C}_{A(BC)}).
\end{equation}
Thus,
\begin{eqnarray}
S(\rho_{A}) 
            &=& {\mathcal E}(\mathcal{C}_{A(BC)})\nonumber\\
            &=& {\mathcal E}( \sqrt{\mathcal{C}^{2}_{AB}+ {\mathcal{C}^{a}_{AC}}^2}~)\nonumber\\
            &\leq & {\mathcal E}(\mathcal{C}_{AB}) +  {\mathcal E}(\mathcal{C}^{a}_{AC})\nonumber\\
            &\leq & E_f(\rho_{AB}) + E_a(\rho_{AC}),
\end{eqnarray}
where the first inequality is by Lemma~\ref{Lem: E}, and the
second inequality is by Eq.~(\ref{eoacoa}).
\end{proof}
Thus, the polygamy relation of distributed entanglement in tripartite quantum
systems obtained in Corollary~\ref{Cor: dual} can have a more
tight form in three-qubit systems. Furthermore, the result of
Theorem~\ref{thm2} together with Eqs.~(\ref{KWmain1}) and
(\ref{puremain1}) give us the following corollary.
\begin{Cor}
For any two-qubit mixed state $\rho_{AB}$ with rank less than or equal to two,
\begin{equation}
 I^\leftarrow(\rho_{AB})\leq E_a(\rho_{AB}),
\label{ubound1}
\end{equation}
\begin{equation}
 E_u^{\leftarrow}(\rho_{AB})\leq E_f(\rho_{AB}).
\label{ubound2}
\end{equation}
\label{Cor: 1}
\end{Cor}

\begin{Rmk}\label{Rmk: 1}
Eq.~(\ref{ubound2}) of Corollary~\ref{Cor: 1} implies that any
two-qubit separable state $\rho_{AB}$ of rank less than or equal
to two has zero UE, $E_u^{\leftarrow}(\rho_{AB})=0$. However,
this is not generally true for two-qubit separable states of rank
larger than two. Here, we provide an example of two-qubit
rank-three separable state with non-zero UE.

{\it Example:} Let us consider the following state in $\C^2\otimes
\C^2\otimes \C^3$ quantum system~\cite{CCJKKL},
\begin{eqnarray}
\ket{\Psi}_{ABC}= \frac{1}{\sqrt{2}}\ket{x}_{AC}\ket{0}_B
+\frac{1}{\sqrt{2}}\ket{y}_{AC}\ket{1}_B, \label{eq:ex_Psi}
\end{eqnarray}
where $\ket{x}$ and $\ket{y}$ are two orthogonal states in the
$\C^2\otimes \C^3$ such that
\begin{align}
\ket{x}=&(\ket{02}+\sqrt{2}\ket{10})/{\sqrt{3}},\nonumber\\
\ket{y}=&(\ket{12}+\sqrt{2}\ket{01})/{\sqrt{3}}.
\end{align}

First, since $\rho_{A}=(\ket{0}_{A}\bra{0}+\ket{1}_{A}\bra{1})/2$,
it is clear that $\mathcal{C}_{A(BC)}=\sqrt{2(1-\T\rho^2_A)}=1$,
therefore we have $S(\rho_A)=1$.

Since $\rho_{AC}=(\ket{x}_{AC}\bra{x}+\ket{y}_{AC}\bra{y})/2$,
Hughston-Jozsa-Wootters (HJW) theorem~\cite{HJW} says that for any
decompositions of
$\rho_{AC}=\sum_{i}p_{i}\ket{\phi_{i}}_{AC}\bra{\phi_{i}}$, there
exists an unitary operator $(u_{ij})$ such that
$\sqrt{p_{i}}\ket{\phi_{i}}_{AC}=(u_{i1}\ket{x}_{AC}+u_{i2}\ket{y}_{AC})/\sqrt{2}$
with $2p_{i}=|u_{i1}|^{2}+|u_{i2}|^{2}$. Thus,
\begin{align}
\T_{C}(\ket{\phi_{i}}_{AC}\bra{\phi_{i}})=&\frac{1}{6p_i}
\begin{pmatrix}
  |u_{i1}|^{2}+2|u_{i2}|^{2} & u_{i1}u_{i2}^{*} \\
  u_{i2}u_{i1}^{*} & |u_{i2}|^{2}+2|u_{i1}|^{2}
\end{pmatrix}
\nonumber\\
=&\frac{1}{3}I_{A}+\frac{1}{3}\ket{\psi_{i}}_{A}\bra{\psi_{i}},
\label{eq:decomp_matrix}
\end{align}
with
$\ket{\psi_{i}}=(u_{i2}^*\ket{0}+u_{i1}^*\ket{1})/{\sqrt{2p_{i}}}$,
and we obtain that
$\mathcal{C}(\ket{\phi_{i}}_{AC})=\frac{2\sqrt{2}}{3}$ for any
pure state $\ket{\phi_{i}}_{AC}$ in any pure state decomposition
of $\rho_{AC}$.

Since $\ket{\phi_{i}}_{AC}$ is a $2\otimes 3$ pure state, we have
\begin{eqnarray}
 E_f(\ket{\phi_{i}}_{AC})&=&{\mathcal E}(\mathcal{C}(\ket{\phi_{i}}_{AC}))\nonumber\\
                         &=& H\Bigl( \frac{2}{3} \Bigr)\nonumber\\
                         &=& \log_2 3 -\frac{2}{3},
\end{eqnarray}
and thus $E_f(\rho_{AC})=E_a(\rho_{AC})=\log_2 3 -\frac{2}{3}$.

Now, we have
$E_u^{\leftarrow}(\rho_{AB})=S(\rho_A)-E_a(\rho_{AC})=\frac{5}{3}-\log_2
3 > 0$, whereas, it can be easily seen that $\rho_{AB}$ has a {\em
Positive Partial Transposition}~(PPT) which is equivalent to
separability for two-qubit states~\cite{horo1}. Thus, $\rho_{AB}$
is a two-qubit, rank-three separable state with non-zero
UE.
\end{Rmk}

\subsection{$n$-qubit systems}
\label{subsec: n-qubit} The polygamy inequality of
entanglement in $n$-qubit systems in Eq.~(\ref{ndual}) gives us an
inequality
\begin{equation}
\mathcal{C}_{A_1 (A_2 \cdots A_n)}  \leq \sqrt{
(\mathcal{C}^a_{A_1 A_2})^2 +\cdots+(\mathcal{C}^a_{A_1 A_n})^2 }.
\label{ndual2}
\end{equation}
Thus, together with Lemma~\ref{Lem: E}, we have the following
theorem.
\begin{Thm}
For any $n$-qubit pure state $\ket{\psi}_{A_1 \cdots A_n}$,
\begin{equation}
S(\rho_{A_1})  \leq  E_a(\rho_{A_1 A_2}) +\cdots+E_a(\rho_{A_1
A_n}). \label{nEdua}
\end{equation}
\label{thm: dual}
\end{Thm}

\begin{proof}
First, let us assume that $ (\mathcal{C}^a_{A_1 A_2})^2
+\cdots+(\mathcal{C}^a_{A_1 A_n})^2 \leq 1$, then we have
\begin{align}
S(\rho_{A_1})=&{\mathcal E}(\mathcal{C}_{A_1 (A_2 \cdots A_n)})\nonumber\\
\leq &{\mathcal E}\left(\sqrt{
(\mathcal{C}^a_{A_1 A_2})^2 +\cdots+(\mathcal{C}^a_{A_1 A_n})^2 }\right)\nonumber\\
\leq &{\mathcal E}\left( \mathcal{C}^a_{A_1 A_2}\right) +
{\mathcal E}\left(\sqrt{(\mathcal{C}^a_{A_1 A_3})^2
+\cdots+(\mathcal{C}^a_{A_1 A_n})^2}\right)\nonumber\\
\leq &{\mathcal E}\left( \mathcal{C}^a_{A_1 A_2}\right)+ {\mathcal
E}\left(\mathcal{C}^a_{A_1 A_3}\right) +\cdots + {\mathcal
E}\left(
\mathcal{C}^a_{A_1 A_n}\right)\nonumber\\
\leq& E_a\left(\rho_{A_1 A_2}\right)
+\cdots+E_a\left(\rho_{A_1 A_n}\right),
\end{align}
where the first inequality is due to the monotonicity of the function
$\mathcal E$, the second and third inequalities are obtained by
iterating Lemma~\ref{Lem: E}, and the last inequality is by Eq.~(\ref{eoacoa}).

Now, assume that $ (\mathcal{C}^a_{A_1 A_2})^2
+\cdots+(\mathcal{C}^a_{A_1 A_n})^2 > 1$. Since $S(\rho_{A_1})\leq
1$ for any $n$-qubit pure state $\ket{\psi}_{A_1 \cdots A_n}$, it
is enough to show that $E_a(\rho_{A_1 A_2}) +\cdots+E_a(\rho_{A_1
A_n}) \geq 1$.

Let us first note that there exist $k \in \{2,\ldots ,n-1 \}$ that
satisfies
\begin{align}
&(\mathcal{C}^a_{A_1 A_2})^2 +\cdots+(\mathcal{C}^a_{A_1 A_k})^2
\leq 1,\nonumber\\
&(\mathcal{C}^a_{A_1 A_2})^2 +\cdots+(\mathcal{C}^a_{A_1
A_{k+1}})^2 >1,
\end{align}
and let
\begin{equation}
T=(\mathcal{C}^a_{A_1 A_2})^2 +\cdots+(\mathcal{C}^a_{A_1
A_{k+1}})^2-1.
\end{equation}

Since, $(\mathcal{C}^a_{A_1 A_{k+1}})^2-T=1-(\mathcal{C}^a_{A_1
A_2})^2 -\cdots-(\mathcal{C}^a_{A_1 A_{k}})^2 $, we have
\begin{equation}
0\leq (\mathcal{C}^a_{A_1 A_{k+1}})^2-T \leq 1,
\end{equation}
and
\begin{align}
1 =& \mathcal E \left( \sqrt{(\mathcal{C}^a_{A_1 A_2})^2
+\cdots+(\mathcal{C}^a_{A_1
A_{k+1}})^2-T} \right)\nonumber\\
\leq& \mathcal E \left( \sqrt{(\mathcal{C}^a_{A_1 A_2})^2
+\cdots+(\mathcal{C}^a_{A_1 A_k})^2} \right)\nonumber\\
&~~~~~~~~~~~~~~~~+\mathcal E \left( \sqrt{\left(\mathcal{C}^a_{A_1 A_{k+1}}\right)^2-T} \right)\nonumber\\
\leq& \mathcal E \left( \mathcal{C}^a_{A_1 A_2}\right)+\cdots
+\mathcal E \left( \mathcal{C}^a_{A_1 A_k}\right)+
\mathcal E \left( \mathcal{C}^a_{A_1 A_{k+1}}\right)\nonumber\\
\leq& E_a(\rho_{A_1 A_2})+\cdots + E_a(\rho_{A_1 A_n}),
\end{align}
where the first and second inequalities are by Lemma~\ref{Lem: E} and
by monotonicity of $\mathcal E$, and the last inequality is by Eq.~(\ref{eoacoa}).
\end{proof}

\section{Unlocalizable Entanglement versus Other Measures of
  Correlation}
\label{sec: other} In this section, we provide some properties of
UE concerned with several other correlation
measures. By investigating the relation between UE
and EoF in $2 \otimes 2 \otimes d$ quantum system, we show that
any two-qubit state with zero UE is a separable state. We
also provide a quantitative relation among entropy, localizable
entanglement, and UE for tripartite mixed states.

\subsection{$2\otimes2\otimes d$ pure state}
\label{subsec: 22d}

Let $\ket{\psi}_{ABC}$ be a tripartite pure state in
$\mathbb{C}^2\otimes\mathbb{C}^2\otimes\mathbb{C}^d$.

\begin{Thm}
  For any 2-qubit state $\rho_{AB}$,
\begin{equation}
  E_u^{\leftarrow}(\rho_{AB})=0\Longrightarrow  \rho_{AB}\textrm{ separable},
\end{equation}
or, equivalently, if $S(\rho_A)= E_a(\rho_{AC})$, then
$E_f(\rho_{AB})=0$.
\label{thm: 22d}
\end{Thm}

\begin{proof}
Suppose $S(\rho_A)= E_a(\rho_{AC})$, and let
$\rho_{AC}=\sum_{i}p_i\ket{\phi^i}_{AC}\bra{\phi^i}$ be an optimal
decomposition such that,
\begin{eqnarray}
E_a(\rho_{AC})&=&\sum_{i}p_i E(\ket{\phi^i}_{AC})\nonumber\\
              &=&\sum_{i}p_i S(\rho^{i}_A),
\end{eqnarray}
where $\rho^{i}_{A}=\T_{C}(\ket{\phi^i}_{AC}\bra{\phi^i})$ and $\rho_A = \sum_{i}p_i \rho^i_A$.

The concavity of von Neumann entropy says, $S(\rho_{A}) \geq
\sum_{i}p_i S(\rho^{i}_A)$ and the equality holds if and only if
$\rho^i_A$ are  identical for all $i$.  So, by the assumption,
$\rho^i_A$ are  identical for all $i$.

Since $\ket{\phi^i}_{AC}$ is a pure state in $2\otimes d$ system
and its concurrence is
$\mathcal{C}(\ket{\phi^i}_{AC})=2\sqrt{det\rho^i_A}$, we also have
that $\mathcal{C}(\ket{\phi^i}_{AC})$ are identical for all $i$,
say $\mathcal{C}(\ket{\phi^i}_{AC})=\mathcal{C}^*_{AC}$.

Now, we have
\begin{equation}
S(\rho^{i}_A)=E_f(\ket{\phi^i}_{AC})=\mathcal E(\mathcal{C}(\ket{\phi^i}_{AC}))=\mathcal E(\mathcal{C}^*_{AC}),
\end{equation}
for all $i$, and
\begin{eqnarray}
\mathcal E(\mathcal{C}_{A(BC)})&=& E_f(\ket{\psi_{A(BC)}})\nonumber\\
                              &=& S(\rho_A)\nonumber\\
                              &=& \sum_{i}p_i S(\rho^{i}_A)\nonumber\\
                              &=& \sum_{i}p_i \mathcal E(\mathcal{C}(\ket{\phi^i}_{AC}))\nonumber\\
                              &=& \mathcal E(\mathcal{C}^*_{AC}),
\end{eqnarray}
where $\mathcal{C}_{A(BC)}$ is the concurrence of
$\ket{\psi}_{ABC}$ between subsystems $A$ and $BC$, and $\mathcal
E(\cdot)$ is the function in Eq.~(\ref{eps}).

Since $\mathcal E(\cdot)$ is strictly monotone increasing, (the
first derivative $\frac{d}{d\mathcal C}\mathcal E(\mathcal C)$ is
0 at $\mathcal C =0$ and positive elsewhere), we have
\begin{equation}
\mathcal{C}_{A(BC)} = \mathcal{C}^*_{AC},
\end{equation}
therefore
\begin{equation}
\mathcal{C}_{A(BC)} \geq \mathcal C^a_{AC} \geq
\sum_{i}p_i\mathcal C(\ket{\phi^i}_{AC}) =
\mathcal{C}^*_{AC}=\mathcal{C}_{A(BC)},
\end{equation}
and thus,
\begin{equation}
\mathcal{C}_{A(BC)} = \mathcal C^a_{AC}.
\end{equation}

Now, by the Theorem~3 in~\cite{CCJKKL}, we have $\mathcal
C(\rho_{AB})=0$ where $\rho_{AB}$ is a 2-qubit state, which
implies $E_f(\rho_{AB})=0.$
\end{proof}

Any two-qubit state with zero UE is separable by
Theorem~\ref{thm: 22d}, and any two-qubit separable state with
rank less than or equal to two has zero UE by
Corollary~\ref{Cor: 1}. However, the converse of Theorem~\ref{thm:
22d} is not generally true, since Remark~\ref{Rmk: 1} provides us a
two-qubit separable state with non-zero UE.

\subsection{Tripartite Mixed State}
\label{sec: tri mixed}
Since it is known that the EoA is not a bipartite measure nor an {\em
  entanglement monotone}~\cite{gs}, it is not clear yet if there is
any quantitative relation between $E_a(\rho_{AB})$ and
$E_a(\rho_{A(BD)})$ for a tripartite mixed state $\rho_{ABD}$. In
fact, this is equivalent to the quantitative relation between
$E_u^{\leftarrow}(\rho_{AC})$ and
$E_u^{\leftarrow}(\rho_{A(CD)})$. This is because, if we consider a
purification $\ket{\psi}_{ABCD}$ of $\rho_{ABD}$, then any direction
of a quantitative relation between $E_a(\rho_{AB})$ and
$E_a(\rho_{A(BD)})$, say $E_a(\rho_{AB}) \leq E_a(\rho_{A(BD)})$,
would give us
\begin{align}
S(\rho_A)=&E_u^{\leftarrow}(\rho_{A(CD))})+E_a(\rho_{AB})\nonumber\\
=&E_u^{\leftarrow}(\rho_{AC)})+E_a(\rho_{A(BD)})\nonumber\\
\geq& E_u^{\leftarrow}(\rho_{AC)})+E_a(\rho_{AB}),
\end{align}
which implies $E_u^{\leftarrow}(\rho_{AC})\leq E_u^{\leftarrow}(\rho_{A(CD)})$.

In this section, we pay our attentions only to local rank-1
measurements of each subsystems, and we derive a quantitative
relation between localizable entanglement, and UE for
tripartite mixed states.

For $\rho_{ACD}=\T_B \ket{\psi}_{ABCD}\bra{\psi}$, let us define
\begin{equation}
\begin{split}
\widetilde{E}_u^{\leftarrow}(\rho_{A(CD)}) &=
\min_{\{M_x\otimes N_y\}} \left[S(\rho_A)-\sum_{x, y} p_{xy} S(\rho^{xy}_A)\right],\\
\end{split}
\label{unlocal3}
\end{equation}
where $p_{xy}\equiv \T[(I_A\otimes M_x\otimes N_y )\rho_{ACD}]$
is the probability of the outcome $x$ and $y$ on subsystems $C$ and $D$ respectively,
and $\rho^{xy}_A\equiv \T_{CD}[(I_A\otimes M_x\otimes N_y)\rho_{ACD}]/p_{xy}$ is the state
of system $A$ when the outcome were $x$ and $y$.
The minimum in Eq.~(\ref{unlocal3}) is taken over all possible rank-1 measurements
$\{M_x\}$ and $\{N_y\}$ on subsystems $C$ and $D$ respectively.
By definition, we have
\begin{equation}
\widetilde{E}_u^{\leftarrow}(\rho_{A(CD)}) \geq E_u^{\leftarrow}(\rho_{A(CD)}).
\end{equation}
Furthermore, we have the following lemma.
\begin{Lem}
For any tripartite state $\rho_{ACD}$,
\begin{equation}
\widetilde{E}_u^{\leftarrow}(\rho_{A(CD)}) \geq E_u^{\leftarrow}(\rho_{AC}).
\end{equation}
\label{lem1}
\end{Lem}
\begin{proof}
For $\rho_{ACD}$, let $\{M_x \}$ and $\{N_y\}$ are the optimal
rank-1 measurements of $C$ and $D$ respectively, such that
\begin{equation}
\widetilde{E}_u^{\leftarrow}(\rho_{A(CD)})=S(\rho_A)-\sum_{x, y} p_{xy} S(\rho^{xy}_A).
\end{equation}
Due to the concavity of von Neumann entropy, we have
\begin{align}
\widetilde{E}_u^{\leftarrow}(\rho_{A(CD)})&=S(\rho_A)-\sum_{x,y} p_{xy} S(\rho^{xy}_A)\nonumber\\
&= S(\rho_A)-\sum_{x} p_{x}\left[\sum_{y}\frac{p_{xy}}{p_{x}} S(\rho^{xy}_A)\right]\nonumber\\
&\geq S(\rho_A)-\sum_{x} p_{x} S\left(\sum_{y}\frac{p_{xy}}{p_{x}} \rho^{xy}_A\right)\nonumber\\
&=S(\rho_A)-\sum_{x} p_{x} S(\rho^{x}_A)\nonumber\\
&\geq E_u^{\leftarrow}(\rho_{AC}),
\end{align}
where
\begin{align}
p_{x}&=\T_{AC}\left[(I_{A}\otimes M_{x})\rho_{AC} \right]\nonumber\\
&=\T_{ACD}\left[(I_{A}\otimes M_{x}\otimes I_{D})\rho_{ACD} \right]\nonumber\\
&=\sum_{y} \T_{ACD}\left[(I_{A}\otimes M_{x}\otimes N_{y})\rho_{ACD} \right]\nonumber\\
&=\sum_{y}p_{xy},
\end{align}

\begin{align}
\rho^x_A&=\frac{1}{p_x}\T_{C}\left[(I_{A}\otimes M_{x})\rho_{AC} \right]\nonumber\\
&=\frac{1}{p_x}\T_{CD}\left[(I_{A}\otimes M_{x}\otimes I_{D})\rho_{ACD} \right]\nonumber\\
&=\sum_{y}\frac{1}{p_x} \T_{CD}\left[(I_{A}\otimes M_{x}\otimes N_{y})\rho_{ACD} \right]\nonumber\\
&=\sum_{y}\frac{p_{xy}}{p_x}\rho^{xy}_A,
\end{align}
and the second inequality is due to the definition of $E_u^{\leftarrow}(\rho_{AC})$.
\end{proof}

Now, we are ready to have the following theorem.

\begin{Thm}
For any tripartite mixed state $\rho_{ABC}$ with a purification $\ket{\psi}_{ABCD}$,
\begin{align}
S(\rho_A)&\geq \widetilde{E}_a(\rho_{AB})+E_u^{\leftarrow}(\rho_{AC}),
\end{align}
where $\widetilde{E}_a(\rho_{AB})$ is the
localizable entanglement~{\em \cite{PVMC}} of $\rho_{AB}$, defined by
\begin{equation}
\widetilde{E}_a(\rho_{AB})=\max_{\{M_x\otimes N_y\}} \sum_{x, y} p_{xy} S(\rho^{xy}_A)
\end{equation}
over all possible rank-1 measurements
$\{M_x\}$ and $\{N_y\}$ on subsystems $C$ and $D$ respectively.
\label{Thm: 3mix}
\end{Thm}
\begin{proof}

Eq.~(\ref{unlocal3}) can be rewritten as
\begin{align}
\widetilde{E}_u^{\leftarrow}(\rho_{A(CD)})&=S(\rho_A)-\max_{\{M_x\otimes N_{y} \}}
\sum_{x,y} p_{xy} S(\rho^{xy}_A)\nonumber\\
&=S(\rho_A)-\widetilde{E}_a(\rho_{AB}),
\label{unlocal4}
\end{align}
and Lemma~\ref{lem1} completes the proof.
\end{proof}

Theorem~\ref{Thm: 3mix} can be considered as an alternative of
Lemma~\ref{Lem: puremain1} for mixed states case. Furthermore,
Theorem~\ref{Thm: 3mix} together with Lemma~\ref{Lem: puremain1}
give us the following simple corollary.

\begin{Cor}
For any tripartite mixed state $\rho_{ABC}$ with a purification $\ket{\psi}_{ABCD}$,
\begin{equation}
E_a(\rho_{A(BC)})\geq \widetilde{E}_a(\rho_{AB}).
\end{equation}
\label{cor2}
\end{Cor}
\begin{proof}
By Theorem~\ref{Thm: 3mix}, we have
\begin{equation}
S(\rho_A) \geq \widetilde{E}_a(\rho_{AB})+E_u^{\leftarrow}(\rho_{AD})
\end{equation}
for any pure state $\ket{\psi}_{ABCD}$, whereas
\begin{equation}
S(\rho_A)=E_a(\rho_{A(BC)})+E_u^{\leftarrow}(\rho_{AD}),
\end{equation}
for the tripartite partition $A-BC-D$.
\end{proof}

\section{Conclusion}
\label{sec: Conclusion}

In this paper, we have proposed the concept of UE, and shown that the
polygamous nature of distributed quantum entanglement in multipartite systems is
strongly due to this unlocalizable character. As the mathematical
interpretation for this polygamous nature of quantum entanglement, we
have established polygamy inequalities of entanglement in tripartite
quantum systems with arbitrary dimension, and multi-qubit systems.  We
have also provided several trade offs between UE and other
correlations such as EoA, and localizable entanglement.

This is the first result where polygamous property of quantum
entanglement in multipartite higher-dimensional quantum systems is
provided.  Furthermore, the proposed inequalities are in terms of the
entropic entanglement measures such as entropy of entanglement for
pure states and EoA.  In other words, the proposed polygamy inequalities
of distributed entanglement have been shown in terms of the actual quantification
of entanglement with operational meanings, rather than using other
entanglement measures such as concurrence.

\section*{Acknowledgments}
JSK would like to thank Soojoon Lee for useful discussion, and acknowledges the support
from iCORE, MITACS (QIP project) and US Army.
GG acknowledges financial support from NSERC and MITACS-QIP.



\begin{thebibliography}{1}

\bibitem{tele}
C. H. Bennett, G. Brassard, C. Crepeau, R. Jozsa, A. Peres and W. K. Wootters,
Phys. Rev. Lett. {\bf 70}, 1895 (1993).

\bibitem{qkd1}
C. Bennett and G. Brassard, {\em in Proceedings of IEEE
International Conference on Computers, Systems, and Signal
Processing }(IEEE Press, New York, Bangalore, India, 1984), p.
175-179.

\bibitem{qkd2}
C. H. Bennett, Phys. Rev. Lett. {\bf 68}, 3121 (1992).

\bibitem{CKW}
V.~Coffman, J.~Kundu and W.~K.~Wootters, Phys. Rev. A {\bf 61}, 052306 (2000).

\bibitem{ov}
T. Osborne and F. Verstraete,
Phys. Rev. Lett. {\bf 96}, 220503 (2006).

\bibitem{KW}
M.~Koashi and A.~Winter, Phys. Rev. A {\bf69}(2) 022309 (2004).

\bibitem{T04}
B. M. Terhal, IBM J. Research and Development 48, 71 (2004).

\bibitem{ww}
W.~K.~Wootters, Phys. Rev. Lett. {\bf 80}, 2245 (1998).

\bibitem{d}
D. P. DiVincenzo {em et al.}, Lect. Notes Comput. Sci. {\bf 1509}, 247
 (1999).

\bibitem{cohen}
O. Cohen, Phys. Rev. Lett. {\bf 80}, 2493 (1998).

\bibitem{lve}
T. Laustsen, F. Verstraete and S. J. van Enk, Quantum Inf. Comput.
{\bf 3}, 64 (2003).

\bibitem{gbs}
G. Gour, S. Bandyopadhay and B. C. Sanders, J. Math. Phys. {\bf
48}, 012108 (2007).

\bibitem{DV}
I. Devetak, A. Winter,
IEEE Transactions on Information Theory {\bf 50}(12) pp. 3183-3196
(2004).

\bibitem{Henderson-Vedral01}
L.~Henderson and V.~Vedral, J. Phys. A: Math. Gen. {\bf 34}, 6899
(2001).

\bibitem{bdsw}
C. H. Bennett, D. P. DiVincenzo, J. A. Smolin and W. K. Wootters,
Phys. Rev. A {\bf 54}, 3824 (1996).

\bibitem{Smolin-Ver-Win} J.~A.~Smolin, F.~Verstraete and A.~Winter,
Phys.~Rev.~A~{\bf 72}, 052317 (2005).

\bibitem{christandl} M.~Christandl and A.~Winter, IEEE Trans. Inf.
  Theory {\bf 51}, 3159--3165 (2005).

\bibitem{nc}
M. A. Nielsen and I. L. Chuang, {\em Quantum Computation and
Quantum Information} (Cambridge University Press, Cambridge, U.K.,
2000).

\bibitem{ys}
C-s. Yu and H-s. Song, Phys. Rev. A {\bf 76}, 022324 (2007).

\bibitem{CCJKKL}
D.~P.~Chi, J.~W.~Choi, K.~Jeong, J.~S.~Kim, T.~Kim and S. Lee, J. Math. Phys. {\bf
49}, 112102 (2008).

\bibitem{HJW}
L. P. Hughston, R. Jozsa and W. K. Wootters,
Phys. Lett. A {\bf 183}, 14 (1993).

\bibitem{horo1} P.~Horodecki,
Phys. Lett. A {\bf 232}, 333 (1997).

\bibitem{gs}
G. Gour and R. W. Spekkens,
Phys. Rev. A {\bf 73}, 062331 (2006).

\bibitem{PVMC}
M.~Popp, F~Verstraete, M.~A.~Martin-Delgado and J.~I.~Cirac, Phys. Rev. A {\bf 71}, 042306 (2005).

\end{thebibliography}
\end{document}